\documentclass[
	conference,
	10pt
]{IEEEtran}

\IEEEoverridecommandlockouts 
\makeatletter
\def\markboth#1#2{\def\leftmark{\@IEEEcompsoconly{\sffamily}\MakeUppercase{\protect#1}}%
\def\rightmark{\@IEEEcompsoconly{\sffamily}\MakeUppercase{\protect#2}}}
\makeatother

\usepackage[amssymb]{SIunits}
\usepackage[cmex10]{amsmath} 
\usepackage{amssymb, amsthm}
\usepackage{bbm} 
\usepackage[latin1]{inputenc}
\usepackage{xcolor}
\usepackage{cite}
\usepackage[pdftex]{graphicx}
\usepackage[caption=false, font=footnotesize]{subfig}
\usepackage{url}
\usepackage{mathtools}
\DeclarePairedDelimiter{\ceil}{\lceil}{\rceil}
\DeclarePairedDelimiter{\floor}{\lfloor}{\rfloor}
\usepackage{xspace} 

\usepackage{booktabs} 

\usepackage{array}
\newcolumntype{L}[1]{>{\raggedright\let\newline\\\arraybackslash\hspace{0pt}}m{#1}}
\newcolumntype{C}[1]{>{\centering\let\newline\\\arraybackslash\hspace{0pt}}m{#1}}
\newcolumntype{R}[1]{>{\raggedleft\let\newline\\\arraybackslash\hspace{0pt}}m{#1}}

\usepackage[
	shortcuts,
	nonumberlist,	
	acronym, 
	hyperfirst=true,
	section=section, 
	order=letter,
	numberedsection=nolabel 
]{glossaries}

\theoremstyle{definition}
\newtheorem{definition}{Definition}

\theoremstyle{plain}

\newtheorem{theorem}{Theorem}

\theoremstyle{remark} 
\newtheorem{remark}{Remark}
\newtheorem{example}{Example}

\newcommand\xqed[1]{%
\leavevmode\unskip\penalty9999 \hbox{}\nobreak\hfill
\quad\hbox{#1}}
\newcommand\demo{\xqed{$\triangle$}}

\newcommand{\define}{\triangleq}
\newcommand{\D}{\mathrm{d}}

\newcommand{\CNnum}{m}

\newcommand{\VNs}{\glspl{vn}\xspace}
\newcommand{\CNs}{\glspl{cn}\xspace}

\newcommand{\nbch}{\ensuremath{n}}

\newcommand{\Indicator}[1]{\mathbbm{1}\left\{#1\right\}}

\newcommand{\tee}{\ensuremath{t}}

\newcommand{\tmax}{\ensuremath{\tee_{\text{max}}}}

\newcommand{\transpose}{\intercal}

\newcommand{\Loss}{\ensuremath{\mathcal{L}}}

\newcommand{\PCs}{\glspl{pc}\xspace}
\newcommand{\GPCs}{\glspl{gpc}\xspace}

\newcommand{\cp}{\ensuremath{\bar{c}_\textup{p}}}
\newcommand{\cthr}{\ensuremath{\bar{c}}}
\newcommand{\vect}[1]{\ensuremath{\boldsymbol{#1}}}
\newcommand{\mat}[1]{\ensuremath{\boldsymbol{#1}}}
\newcommand{\etab}{\ensuremath{\boldsymbol{\eta}}}

\newif\ifshow
\showfalse

\newcommand{\abbr}[1]{{#1}}				

\makeatletter
\let\aclOLD=\acl
\renewcommand{\acl}[1]{%
  \begingroup    
  \let\@@underline=\relax
  \aclOLD{#1}%
  \endgroup
}
\makeatother

\newcommand{\NewA}[3]{
	\newacronym{#1}{#2}{#3}
}

\NewA{af}{AF}{\abbr{a}mplify-and-\abbr{f}orward}
\NewA{apsk}{APSK}{\abbr{a}mplitude \abbr{p}hase-\abbr{s}hift \abbr{k}eying}
\NewA{ask}{ASK}{\abbr{a}mplitude-\abbr{s}hift \abbr{k}eying}
\NewA{ase}{ASE}{\abbr{a}mplified \abbr{s}pontaneous \abbr{e}mission}
\NewA{awgn}{AWGN}{\abbr{a}dditive \abbr{w}hite \abbr{G}aussian \abbr{n}oise}
\NewA{biawgn}{BI-AWGN}{binary-input \abbr{a}dditive \abbr{w}hite \abbr{G}aussian \abbr{n}oise}
\NewA{bep}{BEP}{\abbr{b}it \abbr{e}rror \abbr{p}robability}
\NewA{cep}{CEP}{\abbr{c}odeword \abbr{e}rror \abbr{p}robability}
\NewA{ber}{BER}{\abbr{b}it \abbr{e}rror \abbr{r}ate}
\NewA{qap}{QAP}{\abbr{q}uadratic \abbr{a}assignment \abbr{p}roblem}
\NewA{bicm}{BICM}{\abbr{b}it-\abbr{i}nterleaved \abbr{c}oded \abbr{m}odulation}				
\NewA{cm}{CM}{\abbr{c}oded \abbr{m}odulation}				
\NewA{qpsk}{QPSK}{quadrature phase-shift keying}				
\NewA{mlcm}{MLCM}{multilevel coded modulation}				
\NewA{bpsk}{BPSK}{\abbr{b}inary \abbr{p}hase-\abbr{s}hift \abbr{k}eying}				
\NewA{bsc}{BSC}{\abbr{b}inary \abbr{s}ymmetric \abbr{c}hannel}				
\NewA{brgc}{BRGC}{\abbr{b}inary \abbr{r}eflected \abbr{G}ray \abbr{c}ode}				
\NewA{cf}{CF}{\abbr{c}haracteristic \abbr{f}unction}
\NewA{csit}{CSIT}{\abbr{c}annnel \abbr{s}tate \abbr{i}nformation at the \abbr{transmitter}}
\NewA{csi}{CSI}{\abbr{c}annnel \abbr{s}tate \abbr{i}nformation}
\NewA{df}{DF}{\abbr{d}ecode-and-\abbr{f}orward}
\NewA{fd}{FD}{\abbr{f}ull-\abbr{d}uplex}
\NewA{fft}{FFT}{\abbr{f}ast \abbr{F}ourier \abbr{t}ransform}
\NewA{iid}{IID}{independent and \abbr{i}dentically \abbr{d}istributed}
\NewA{isi}{ISI}{\abbr{i}nter\abbr{s}ymbol \abbr{i}nterference}
\NewA{lb}{LB}{\abbr{l}ower \abbr{b}ound}
\NewA{map}{MAP}{\abbr{m}aximum \abbr{a} \abbr{p}osteriori}
\NewA{mf}{MF}{\abbr{m}odulo-and-\abbr{f}orward}
\NewA{mlan}{MLAN}{\abbr{m}odulo-\abbr{l}attice \abbr{a}dditive \abbr{n}oise}
\NewA{ml}{ML}{\abbr{m}aximum \abbr{l}ikelihood}
\NewA{mmse}{MMSE}{\abbr{m}inimum \abbr{m}ean \abbr{s}quare \abbr{e}rror}
\NewA{nlpc}{NLPC}{\abbr{n}onlinear \abbr{p}hase \abbr{c}ompensation}
\NewA{nlpn}{NLPN}{\abbr{n}onlinear \abbr{p}hase \abbr{n}oise}
\NewA{nc}{NC}{\abbr{n}etwork \abbr{c}oding}
\NewA{ofmd}{OFDM}{\abbr{o}rthogonal \abbr{f}requency-\abbr{d}ivision \abbr{m}ultiplexing}			
\NewA{dp}{DP}{\abbr{d}ual-\abbr{p}olarization}
\NewA{pam}{PAM}{\abbr{p}ulse \abbr{a}mplitude \abbr{m}odulation}
\NewA{pdf}{PDF}{\abbr{p}robability \abbr{d}ensity \abbr{f}unction}
\NewA{plnc}{PNC}{\abbr{p}hysical-layer \abbr{n}etwork \abbr{c}oding}
\NewA{psk}{PSK}{\abbr{p}hase-\abbr{s}hift \abbr{k}eying}
\NewA{pmd}{PMD}{\abbr{p}olarization \abbr{m}ode \abbr{d}ispersion}
\NewA{pm}{PM}{\abbr{p}olarization-\abbr{m}ultiplexed}
\NewA{pdm}{PDM}{\abbr{p}olarization-\abbr{d}ivision \abbr{m}ultiplexing}
\NewA{qam}{QAM}{\abbr{q}uadrature \abbr{a}mplitude \abbr{m}odulation}
\NewA{sqp}{SQP}{\abbr{s}equential \abbr{q}uadratic \abbr{p}rogramming}
\NewA{rd}{RD}{\abbr{r}adii \abbr{d}istribution}
\NewA{sep}{SEP}{\abbr{s}ymbol \abbr{e}rror \abbr{p}robability}
\NewA{ser}{SER}{\abbr{s}ymbol \abbr{e}rror \abbr{r}ate}
\NewA{si}{SI}{\abbr{s}ide \abbr{i}nformation}
\NewA{sp}{SP}{\abbr{s}ingle-\abbr{p}olarization}
\NewA{sanr}{SNR}{\abbr{s}ignal-to-(additive-)\abbr{n}oise \abbr{r}atio}
\NewA{snr}{SNR}{\abbr{s}ignal-to-\abbr{n}oise \abbr{r}atio}
\NewA{snlse}{sNLSE}{\abbr{s}tochastic \abbr{n}onlinear \abbr{S}chr\"odinger \abbr{e}quation}
\NewA{nlse}{NLSE}{\abbr{n}onlinear \abbr{S}chr\"odinger \abbr{e}quation}
\NewA{stwrc}{sTRC}{\abbr{s}eparated \abbr{t}wo-way \abbr{r}elay \abbr{c}hannel}
\NewA{stwtrc}{sTTRC}{\abbr{s}eparated \abbr{t}wo-way \abbr{t}wo-\abbr{r}elay \abbr{c}hannel}
\NewA{ts}{TS}{\abbr{t}wo-\abbr{s}tage}
\NewA{twrc}{TRC}{\abbr{t}wo-way \abbr{r}elay \abbr{c}hannel}
\NewA{twtrc}{TTRC}{\abbr{t}wo-way \abbr{t}wo-\abbr{r}elay \abbr{c}hannel}
\NewA{wdm}{WDM}{\abbr{w}avelength-\abbr{d}ivision \abbr{m}ultiplexing}
\NewA{vn}{VN}{variable node}
\NewA{cn}{CN}{constraint node}
\NewA{de}{DE}{density evolution}
\NewA{sc}{SC}{spatially-coupled}
\NewA{ldpc}{LDPC}{low-density parity-check}
\NewA{bp}{BP}{belief propagation}
\NewA{dm}{DM}{dispersion-managed}
\NewA{spm}{SPM}{self-phase modulation}
\NewA{xpm}{XPM}{cross-phase modulation}
\NewA{fwm}{FWM}{four-wave-mixing}
\NewA{ixpm}{IXPM}{intrachannel cross-phase modulation}
\NewA{ifwm}{IFWM}{intrachannel four-wave-mixing}
\NewA{ssfm}{SSFM}{split-step Fourier method}
\NewA{fec}{FEC}{forward error correction}
\NewA{psd}{PSD}{power spectral density}
\NewA{ook}{OOK}{on-off keying}
\NewA{smf}{SMF}{single-mode fiber}
\NewA{ssmf}{SSMF}{standard single-mode fiber}
\NewA{dbp}{DBP}{digital backpropagation}
\NewA{edfa}{EDFA}{erbium-doped fiber amplifier}
\NewA{ofdm}{OFDM}{orthogonal frequency division multiplexing}
\NewA{exit}{EXIT}{extrinsic information transfer}
\NewA{pexit}{P-EXIT}{protograph extrinsic information transfer}
\NewA{osnr}{OSNR}{optical signal-to-noise ratio}
\NewA{roadm}{ROADM}{reconfigurable optical add-drop multiplexer}
\NewA{rps}{RPS}{Raman pump station}
\NewA{mlse}{MLSE}{maximum likelihood sequence estimation}
\NewA{dcf}{DCF}{dispersion compensating fiber}
\NewA{tcm}{TCM}{trellis coded modulation}
\NewA{edc}{EDC}{electronic dispersion compensation}
\NewA{ldpcc}{LDPCC}{low-density parity-check convolutional}
\NewA{llr}{LLR}{log-likelihood ratio}
\NewA{ra}{RA}{repeat-accumulate}
\NewA{ira}{IRA}{irregular-repeat-accumulate}
\NewA{ara}{ARA}{accumulate-repeat-accumulate}
\NewA{mi}{MI}{mutual information}
\NewA{vdmm}{VDMM}{variable degree matched mapping}
\NewA{gmi}{GMI}{generalized mutual information}
\NewA{wd}{WD}{windowed decoder}
\NewA{gn}{GN}{Gaussian noise}
\NewA{hd}{HD}{hard-decision}
\NewA{hdd}{HDD}{hard-decision decoding}
\NewA{sd}{SD}{soft-decision}
\NewA{sdd}{SDD}{soft-decision decoding}
\NewA{emp}{EMP}{extrinsic message passing}
\NewA{imp}{IMP}{intrinsic message passing}
\NewA{gldpc}{GLDPC}{generalized low-density parity-check}
\NewA{scgldpc}{SC-GLDPC}{spatially-coupled generalized low-density parity-check}
\NewA{scldpc}{SC-LDPC}{spatially-coupled low-density parity-check}
\NewA{scfec}{SC-FEC}{spatially-coupled forward error correction}
\NewA{tbbc}{TBBC}{tightly-braided block code}
\NewA{gpc}{GPC}{generalized product code}
\NewA{otn}{OTN}{optical transport network}
\NewA{bch}{BCH}{Bose--Chaudhuri--Hocquenghem}
\NewA{rs}{RS}{Reed--Solomon}
\NewA{bbbc}{BBBC}{block-wise braided block code}
\NewA{hpc}{HPC}{half-product code}
\NewA{pc}{PC}{product code}
\NewA{rv}{RV}{random variable}
\NewA{mbp}{MBP}{multiple block permutator}
\NewA{cer}{CER}{codeword error probability}
\NewA{oh}{OH}{overhead}
\NewA{bdd}{BDD}{bounded-distance decoding}
\NewA{ncg}{NCG}{net coding gain}
\NewA{gwbp}{Galton--Watson branching process}{Galton--Watson branching process}

\newacronym[%
	longplural={binary erasure channels},%
	shortplural={BECs}%
]{bec}{BEC}{binary erasure channel}%

\begin{document}

\title{Deterministic and Ensemble-Based
Spatially-Coupled Product Codes}


\author{
	\IEEEauthorblockN{
	Christian Häger\IEEEauthorrefmark{2},
	Henry D.~Pfister\IEEEauthorrefmark{3},
	Alexandre Graell i Amat\IEEEauthorrefmark{2}, and
	Fredrik Brännström\IEEEauthorrefmark{2} 
	\thanks{This work was partially funded by the Swedish Research
	Council under grant \#2011-5961. The work of H.~Pfister was
	supported in part by the National Science Foundation (NSF) under
	Grant No.~1320924. Any opinions, findings, conclusions, and
	recommendations expressed in this material are those of the authors
	and do not necessarily reflect the views of the NSF.  }}

	\IEEEauthorblockA{\IEEEauthorrefmark{2}%
	Department of Signals and Systems,
	Chalmers University of Technology,
	Gothenburg, Sweden}
	\IEEEauthorblockA{\IEEEauthorrefmark{3}%
	Department of Electrical and Computer Engineering, Duke University,
	Durham, North Carolina
	}
}

\maketitle


\begin{abstract}
Several authors have proposed spatially-coupled (or
convolutional-like) variants of \PCs. In this paper, we focus on a
parametrized family of generalized \PCs that recovers some of these
codes (e.g., staircase and block-wise braided codes) as special cases
and study the iterative decoding performance over the \acl{bec}. Even
though our code construction is deterministic (and not based on a
randomized ensemble), we show that it is still possible to rigorously
derive the \gls{de} equations that govern the asymptotic performance.
The obtained \gls{de} equations are then compared to those for a
related spatially-coupled \gls{pc} ensemble. In particular, we show
that there exists a family of (deterministic) braided codes that
follows the same DE equation as the ensemble, for any spatial length
and coupling width. 

\end{abstract}

\glsresetall

\section{Introduction}

Several authors have proposed modifications of the classical \gls{pc}
construction by Elias \cite{Elias1954}, typically by considering
nonrectangular code arrays. These modifications can be regarded as
generalized \gls{ldpc} codes \cite{Tanner1981}, where the underlying
Tanner graph consists exclusively of degree-2 \VNs. We refer to such
codes as generalized PCs (GPCs)\glsunset{gpc}. For example, \GPCs have
been investigated by many authors as practical solutions for
high-speed fiber-optical communications \cite{Justesen2011,
Smith2012a, Jian2014, Zhang2014, Haeger2015ofc}. 


For the \gls{bec}, we are interested in the asymptotic iterative
decoding performance of \GPCs whose associated code arrays have a
spatially-coupled or convolutional-like structure. Examples include
braided codes \cite{Jian2014,Feltstrom2009} and staircase codes
\cite{Smith2012a}. Spatially-coupled codes have attracted significant
attention in the literature due to their outstanding performance under
iterative decoding \cite{Kudekar2011, Yedla2014}.


 
An asymptotic analysis is typically based on \gls{de} \cite{Luby1998b,
Richardson2001} using an ensemble argument. This approach was taken
for spatially-coupled PCs in \cite{Jian2012, Zhang2015}. However, a
randomly chosen code from these ensembles is unlikely to have a
product array (row-column) structure and hence does not resemble the
codes that are implemented in practice, e.g., staircase or braided
codes. It is thus desirable to make precise statements about the
performance of sequences of deterministic (and structured) \GPCs.  



We consider the high-rate regime, where one assumes that component
codes correct a fixed number of erasures and then studies the case
where the component code length $n$ tends to infinity. Using a
Chernoff bound, one finds that for any fixed erasure probability $p$,
the decoding will fail for large $n$ with high probability.
Therefore, it is customary to let the erasure probability decay slowly
as $c/n$ for some $c > 0$. This leads to rigorous decoding thresholds
in terms of $c$ which may be interpreted as the effective channel
quality. The high-rate regime is also the regime that is relevant in
practice: It is at high rates where \GPCs are competitive compared to
\gls{ldpc} codes and practical \GPCs typically employ long component
codes with small erasure-correcting capability \cite{Justesen2011,
Jian2014, Smith2012a}. 

The main contribution of this paper is to show that, analogous to
\gls{de} for code ensembles, there exists a large class of
deterministic \glspl{gpc} whose asymptotic performance in the
high-rate regime is rigorously characterized in terms of a recursive
DE equation. The code construction we propose here is sufficiently
general to recover (block-wise) braided and staircase codes as special
cases. Our result generalizes previous work in \cite{Justesen2011}
from conventional PCs to a large class of deterministic \GPCs. We
further provide a detailed comparison between deterministic
spatially-coupled PCs and the ensembles in \cite{Jian2012, Zhang2015}
via their respective DE equations. For example, we show that there
exists a family of block-wise braided codes that follows the same DE
recursion as the ensemble in \cite{Jian2012}. This implies that
certain ensemble-properties proved in \cite{Jian2012} also apply to
deterministic \GPCs.  

\emph{Notation.} We use boldface letters to denote vectors and
matrices (e.g., $\vect{x}$ and $\mat{A}$). The symbols $\vect{0}_m$
and $\vect{1}_m$ denote the all-zero and all-one vectors of length
$m$, where the subscript may be omitted. The tail-probability of a
Poisson random variable is defined as $\Psi_{\geq \tee}(x) \define 1 -
\sum_{i=0}^{\tee-1} \Psi_{= i}(x)$, where $\Psi_{= i}(x) \define
\frac{x^i}{i!} e^{-x}$. We use boldface to denote the element-wise
application of a scalar-valued function to a vector, e.g.,
$\boldsymbol{\Psi}_{\geq \tee}(\vect{x})$ applies $\Psi_{\geq
t}(\cdot)$ to each element in $\vect{x}$. For vectors, we use
$\vect{x} \succeq \vect{y}$ if $x_i \geq y_i$ for all $i$. We define
$[m] \define \{1, 2, \dots, m\}$. The indicator function is denoted by
$\Indicator{\cdot}$.

\section{Code Construction and Density Evolution for Deterministic
Generalized Product Codes}

We denote a \gls{gpc} by $\mathcal{C}_n({\etab})$, where $n$ is
proportional to the number of \CNs in the underlying Tanner graph and
$\etab$ is a binary, symmetric $L \times L$ matrix that defines the
graph connectivity.  Recall that \GPCs also have a natural array
representation: The code $\mathcal{C}_n({\etab})$ can alternatively be
defined as the set of all code arrays of a given shape (see
Fig.~\ref{fig:code_arrays} for examples) such that each row and column is a
codeword in some component code.  Thus, one may alternatively think
about $\etab$ as specifying the array shape.  We will see in the
following that different choices for $\etab$ recover well-known code
classes. 

\subsection{Code Construction}
\label{sec:construction}

Let $\gamma > 0$ be some fixed and arbitrary constant such that $d
\define \gamma n$ is an integer. To construct the Tanner graph that
defines $\mathcal{C}_n(\etab)$, assume that there are $L$ positions.
Then, place $d$ \CNs at each position and connect each \gls{cn} at
position $i$ to each \gls{cn} at position $j$ through a \gls{vn} if
and only if $\eta_{i,j} = 1$. 

\begin{example} A \gls{pc} is obtained by choosing $L = 2$ and $\etab
	= \left(\begin{smallmatrix} 0 & 1\\ 1 & 0
\end{smallmatrix}\right)$. The two positions correspond to ``row''
and ``column'' codes. For $\gamma = 1$, the code array
is of size $n \times n$.
\demo
\end{example}

For a fixed $n$, the constant $\gamma$ scales the number of CNs in the
graph. This is inconsequential for the asymptotic analysis (where $n
\to \infty)$ and $\gamma$ manifests itself in the DE equations merely
as a scaling parameter. Its choice will become clear once we discuss
the comparison of codes defined by different $\etab$-matrices in
Sec.~\ref{sec:sc_pcs}. 

\CNs at position $i$ have degree $d \sum_{j \neq i} \eta_{i,j} +
\eta_{i,i}(d - 1)$, where the second term arises from the fact that we
cannot connect a \gls{cn} to itself if $\eta_{i,i} = 1$.  The \gls{cn}
degree specifies the length of the component code associated with the
\gls{cn}. We assume that each \gls{cn} corresponds to a $\tee$-erasure
correcting component code. This assumption is relaxed in
Sec.~\ref{sec:potential_threshold_optimization}.

\subsection{Iterative Decoding}

Suppose that a codeword of $\mathcal{C}_n(\etab)$ is transmitted over
the \gls{bec}\footnote{In practice, \GPCs are used to correct errors
and not erasures. However, the (rigorous) analysis over the \gls{bec}
can be used to closely approximate the performance also over the
binary symmetric channel, see, e.g., \cite{Justesen2011}.} with
erasure probability $p=c/n$ for $c > 0$. The decoding is performed
iteratively assuming $\ell$ iterations of bounded-distance decoding
for the component codes associated with all CNs. Thus, in each
iteration, if the weight of an erasure pattern for a CN is less than
or equal to $\tee$, the pattern is corrected. If the weight exceeds
$\tee$, we say that the component code declares a decoding failure in
that iteration. 

\subsection{Density Evolution}
\label{sec:de}

We wish to characterize the decoding performance in the limit $n \to
\infty$. To that end, assume that we compute 
\begin{align}
	\label{eq:de}
	\!\!\!\!\vect{z}^{(\ell)} = \vect{\Psi}_{\geq \tee+1} (c \mat{B}
	\vect{x}^{(\ell-1)} ) ,  \text{with } 
	\vect{x}^{(\ell)} = \vect{\Psi}_{\geq \tee} (c \mat{B} 
	\vect{x}^{(\ell-1)} ) , 
\end{align}
where $\vect{x}^{(0)} = \vect{1}_L$ and $\mat{B} \define \gamma
\etab$. The main result is as follows. 

\begin{theorem}
	\label{th:main}
	Let the random variable $W$ be the fraction of component codes that
	declare decoding failures in iteration $\ell$. Then, $W$ converges
	almost surely to $\frac{1}{L} \sum_{i=1}^L z_i^{(\ell)}$ as $n \to
	\infty$. 
\end{theorem}

\begin{proof}[Proof (Outline)]
	The decoding can be represented by applying a peeling algorithm to
	the residual graph which is obtained from the Tanner graph by
	deleting known VNs and collapsing erased VNs into edges
	\cite{Justesen2011, Jian2014, Zhang2015}. Our code construction is
	such that the residual graph corresponds to an inhomogeneous random
	graph \cite{Bollobas2007}. The expected value of a suitably defined
	function applied to such a graph converges to the expected value of
	the same function applied to a multi-type Poisson branching process
	\cite{Bollobas2007}. One can show that the peeling algorithm
	constitutes a valid function and that $\frac{1}{L} \sum_{i=1}^L
	z_i^{(\ell)}$ corresponds to its expected value when applied to the
	branching process. Concentration is established by using the method
	of typical bounded differences \cite{Warnke2012}. For a complete
	proof we refer the reader to \cite{Haeger2015tit}. 
\end{proof}


Th.~\ref{th:main} is analogous to the \gls{de} analysis for \gls{ldpc}
codes \cite[Th.~2]{Richardson2001}. For notational convenience, we
define $h(x) \define \Psi_{\geq \tee}(c x)$, so that the recursion in
\eqref{eq:de} can be written as
\begin{align}
	\label{eq:de_gpc}
	\vect{x}^{(\ell)} = \vect{h} ( \mat{B} \vect{x}^{(\ell-1)}).
\end{align}
\begin{definition}
	\label{def:threshold}
The decoding threshold is defined to be 
\begin{align}
	\label{eq:gpc_threshold}
	\cthr \define \sup\{ c \geq 0 \, | \, \vect{x}^{(\infty)} =
	\vect{0}_L
	\}.
\end{align}
\end{definition}

\section{Spatially-Coupled Product Codes}

\subsection{Deterministic Spatially-Coupled Product Codes}
\label{sec:sc_pcs}

We are interested in cases where $\etab$ (and hence $\mat{B}$) has a
band-diagonal ``convolutional-like'' structure. The associated code
can then be classified as a spatially-coupled \gls{pc}.

\begin{figure}[t]
	\vspace{-0.4cm}
	\centering
	\subfloat[staircase code]{\includegraphics{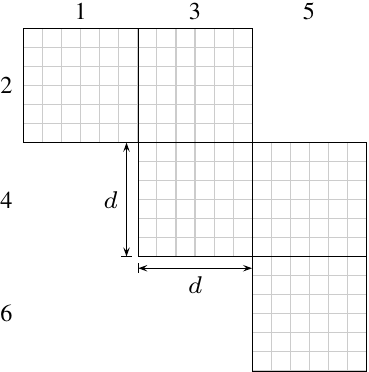}}
	\qquad
	\subfloat[block-wise braided code]{\includegraphics{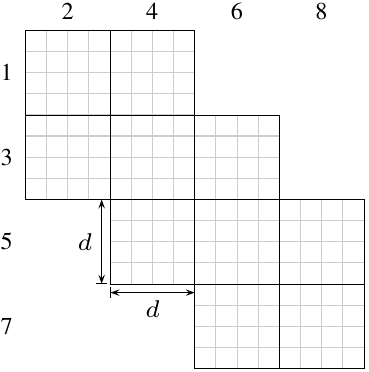}}
	\caption{Code arrays for $\mathcal{C}_{12}(\etab)$, where in (a)
	$\gamma = 1/2$ and in (b) $\gamma = 1/3$. Numbers indicate the position
	indices in the code construction. }
	\vspace{-0.3cm}
	\label{fig:code_arrays}
\end{figure}

\begin{example}
	\label{ex:staircase_codes}
	For $L \geq 2$, the matrix $\etab$ describing a staircase code
	\cite{Smith2012a} has entries $\eta_{i, i+1} = \eta_{i+1,i} = 1$ for
	$i \in [L-1]$ and zeros elsewhere. The corresponding code array is
	shown in Fig.~\ref{fig:code_arrays}(a), where $L = 6$, $n = 12$, and
	$\gamma = 1/2$. \demo
\end{example}

\begin{example}
	\label{ex:braided_codes}
	For even $L\geq 4$, let $\eta_{i,i+1} = \eta_{i+1, i} = 1$ for $i
	\in [L-1]$, $\eta_{2i-1, 2i+2} = \eta_{2i+2, 2i-1} = 1$ for $i \in
	[L/2-1]$, and zeros elsewhere. The resulting matrix $\etab$
	describes a particular instance of a block-wise braided
	code\footnote{We are somewhat liberal in our interpretation of the
	definition in \cite{Feltstrom2009} which is based on multiple block
	permutators. In \cite{Feltstrom2009}, these permutators are linked
	to the dimension of the component code, which turns out to be
	unnecessarily narrow for our purposes.} \cite{Feltstrom2009}.  The
	code array is shown in Fig.~\ref{fig:code_arrays}(b), where $L = 8$,
	$n = 12$, and $\gamma = 1/3$. \demo
\end{example}

The threshold $\cthr$ in Def.~\ref{def:threshold} is a function of
$\etab$ and the scaling parameter $\gamma$. A reasonable scaling to
compare different spatially-coupled PCs is to choose $\gamma$ such
that $\lim_{L \to \infty} \frac{1}{L} \sum_{i=1}^L \sum_{j=1}^L
B_{i,j} = 1 $. For example, $\gamma = 1/2$ and $\gamma = 1/3$ for
staircase and the above braided codes, respectively. This ensures that
in both cases the component codes have length $n$, except at the array
boundaries, see Fig.~\ref{fig:code_arrays}. The matrix $\mat{B}$ is
then referred to as an averaging matrix. 

\subsection{Spatially-Coupled Product Code Ensembles}

We wish to compare the obtained DE recursion in \eqref{eq:de_gpc} to
the \gls{de} recursion for the spatially-coupled PC ensemble defined
in \cite{Jian2012}. We review the necessary background in this
section. 

Let $\mathcal{B}$ be a $\tee$-erasure correcting component code of
length $n$. The Tanner graph corresponding to one particular code in
the spatially-coupled $(\mathcal{B}, \CNnum, L, w)$ ensemble, where
$L$ and $w$ are referred to as the spatial length and coupling width,
respectively, is constructed as follows (cf.~\cite[Def.~2]{Jian2012}).
Place $\CNnum$ degree-$\nbch$ \CNs corresponding to $\mathcal{B}$ at
each position $i \in [L]$ and place $\CNnum\nbch /2$ degree-2 \VNs at
each position $i \in [L']$, where $L' \define L-w+1$.  The
$\CNnum\nbch$ VN and CN sockets at each position are partitioned into
$w$ groups of $\CNnum\nbch/w$ sockets via a uniform random
permutation. Let $\mathcal{S}_{i,j}^{(v)}$ and
$\mathcal{S}_{i,j}^{(c)}$ be, respectively, the $j$-th group for the
VNs and CNs at position $i$, where $j \in [w]$. The Tanner graph is
constructed by connecting $\mathcal{S}_{i,j}^{(v)}$ to
$\mathcal{S}_{i+j,w-j+1}^{(c)}$. 

The ensemble-averaged performance for $\CNnum \to \infty$ is studied
in \cite{Jian2012}.  Without going into the details, the obtained DE
recursion in the high-rate regime (where, additionally, $\nbch \to
\infty$ and $p = c/\nbch$) is given by \cite[eq.~(9)]{Jian2012}
\begin{align}
	\label{eq:de_sc_ensemble}
	\vect{\tilde{x}}^{(\ell)} = c \mat{A} \mat{\Psi}_{\geq \tee}
	(\mat{A}^\transpose \vect{\tilde{x}}^{(\ell-1)}),
\end{align}
where $\vect{\tilde{x}}^{(0)} = c\vect{1}_{L'}$ and $\mat{A}$ is
an $L' \times L$ matrix with
entries 
\begin{align}
	\label{eq:sc_matrix}
	\!\!\!\!A_{i,j} = w^{-1} \Indicator{1 \leq j-i+1 \leq w}, \text{ for } i
	\in [L'], j \in [L]. 
\end{align}

\begin{remark}
	In \cite{Zhang2015}, a modified spatially-coupled PC ensemble is
	considered. The obtained DE recursion is \cite[eq.~(4), $v =
	2$]{Zhang2015}
	\begin{align}
		\vect{y}^{(\ell)} = c \mat{A}^\transpose \mat{A} \mat{\Psi}_{\geq \tee}
		( \vect{y}^{(\ell-1)}),
	\end{align}
	which is identical to
	\eqref{eq:de_sc_ensemble}
	choosing $\vect{\tilde{x}}^{(\ell)} = c \mat{A} \mat{\Psi}_{\geq
	\tee} (\vect{y}^{(\ell)})$. 
\end{remark}

Observe that \eqref{eq:de_sc_ensemble} exhibits a double averaging due
to the randomized edge connections for both VNs and CNs at each
position. Using the substitution $\vect{x}^{(\ell)} = \mat{\Psi}_{\geq
\tee} ( \mat{A} \vect{\tilde{x}}^{(\ell-1)})$ with
$\vect{\tilde{x}}^{(\ell)} = c \mat{A} \mat{\Psi}_{\geq \tee} (
\vect{x}^{(\ell)})$, the recursion becomes
\begin{align}
	\label{eq:de_sc_ensemble_mod}
	\vect{x}^{(\ell)} = \mat{\Psi}_{\geq \tee} (c  \mat{\tilde{B}}
	\vect{x}^{(\ell-1)}) = \vect{h}(\mat{\tilde{B}} \vect{x}^{(\ell-1)}),
\end{align}
where $\vect{x}^{(0)} = \vect{1}_L$ and $\mat{\tilde{B}} \define
\mat{A}^\transpose \mat{A}$ is a symmetric $L \times L$ matrix. For $L
= 6$, the  $\mat{\tilde{B}}$-matrices for $w=2$ and $w = 3$ are,
respectively, given by 
\begin{align}
	\label{eq:b_matrix}
	\!\!\!\frac{1}{4}
	\begin{pmatrix}
		1 & 1 & 0 & 0 & 0 & 0\\
		1 & 2 & 1 & 0 & 0 & 0\\
		0 & 1 & 2 & 1 & 0 & 0\\
		0 & 0 & 1 & 2 & 1 & 0\\
		0 & 0 & 0 & 1 & 2 & 1\\
		0 & 0 & 0 & 0 & 1 & 1\\
	\end{pmatrix}, 
	\quad
	\frac{1}{9}
	\begin{pmatrix}
		1 & 1 & 1 & 0 & 0 & 0\\
		1 & 2 & 2 & 1 & 0 & 0\\
		1 & 2 & 3 & 2 & 1 & 0\\
		0 & 1 & 2 & 3 & 2 & 1\\
		0 & 0 & 1 & 2 & 2 & 1\\
		0 & 0 & 0 & 1 & 1 & 1\\
	\end{pmatrix}.
\end{align}

\section{Comparison of Deterministic and Ensemble-Based Codes}

Comparing the equations, one finds that the ensemble DE recursion
\eqref{eq:de_sc_ensemble_mod} has the same form as \eqref{eq:de_gpc}.
The difference lies in the averaging due to the matrix
$\mat{\tilde{B}}$. 

\begin{example}
	It can be shown that staircase codes are contained in the ensemble
	for $\CNnum = \nbch/2$ and $w = 2$ using a proper choice of
	permutations. It is therefore tempting to conjecture that for $w =
	2$ the recursion \eqref{eq:de_sc_ensemble_mod} also applies to
	staircase codes. However, for staircase codes with $L = 6$, we have
	\begin{align}
		\label{eq:eta_staircase}
	\mat{B} = \gamma \etab = 
	\frac{1}{2}
		\begin{pmatrix}
			0 & 1 & 0 & 0 & 0 & 0\\
			1 & 0 & 1 & 0 & 0 & 0\\
			0 & 1 & 0 & 1 & 0 & 0\\
			0 & 0 & 1 & 0 & 1 & 0\\
			0 & 0 & 0 & 1 & 0 & 1\\
			0 & 0 & 0 & 0 & 1 & 0\\
		\end{pmatrix} ,
	\end{align}
	which is different from the matrix $\mat{\tilde{B}}$ for $w = 2$ in
	\eqref{eq:b_matrix}.
	\demo
\end{example}

\begin{example}
	\label{ex:braided_effective}
	For the braided codes in Ex.~\ref{ex:braided_codes}, one can
	simplify \eqref{eq:de_gpc} by exploiting the inherent symmetry in
	the code construction, which implies $x_i^{(\ell)} =
	x_{i+1}^{(\ell)}$ for odd $i$ and any $\ell$.  It is then sufficient
	to retain odd (or even) positions in \eqref{eq:de_gpc}. With this
	simplification, the effective averaging matrix\footnote{The reader
	may wonder to what code the matrix \eqref{eq:braided_simplified}
	corresponds to, i.e., the code $\mathcal{C}_n(\etab)$ that results
	from using $\etab = 3 \mat{B}'$. One can show that
	$\mathcal{C}_n(\etab)$ can be interpreted as a symmetric subcode of
	the braided code, see \cite{Haeger2015tit, Pfister2015}.} for $L = 12$
	is
\begin{align}
	\label{eq:braided_simplified}
 	\mat{B}' = \frac{1}{3}  
	\begin{pmatrix}
		1 & 1 & 0 & 0 & 0 & 0\\
		1 & 1 & 1 & 0 & 0 & 0\\
		0 & 1 & 1 & 1 & 0 & 0\\
		0 & 0 & 1 & 1 & 1 & 0\\
		0 & 0 & 0 & 1 & 1 & 1\\
		0 & 0 & 0 & 0 & 1 & 1\\
	\end{pmatrix}, 
\end{align}
where $\mat{B}'$ may be used to replace $\mat{B}$ in
\eqref{eq:de_gpc}. Again, one finds that $\mat{B}'$ is different from
the matrices $\mat{\tilde{B}}$ in \eqref{eq:b_matrix}.
\demo
\end{example}

\subsection{Ensemble Performance via Deterministic Codes}

Since $\etab$ is binary, all entries in $\mat{B}$ are either zero or
equal to $\gamma$. To construct spatially-coupled PCs that
follow the same DE recursion as the ensemble, we need to ``emulate''
different multiplicities in the matrix $\mat{B}$. This is done as
follows. 

\begin{definition}
	\label{def:deterministic_ensemble}
For given $L$ and $w$, let $\gamma = 1/w^2$ and $\mat{P} = w^2
\mat{A}^\transpose\mat{A}$, where $\mat{A}$ is defined by
\eqref{eq:sc_matrix}. We define $\etab$ as follows. First, replace
each entry $P_{i,j}$ in $\mat{P}$ by a symmetric $w \times w$ matrix
with $P_{i,j}$ ones in each row and column. The resulting $wL \times
wL$ matrix is denoted by $\etab'$. Finally, $\etab$ is given by
\begin{equation}
\label{eq:eta_prescription}
\begin{aligned}
	\eta_{2i, 2j-1} = \eta_{i,j}', \quad
	\eta_{2i-1, 2j} = \eta_{j,i}', \quad \text{for } i,j \in [wL].
\end{aligned}
\end{equation}
\end{definition}

\begin{figure}[t]
	\centering
	\includegraphics[width=7cm]{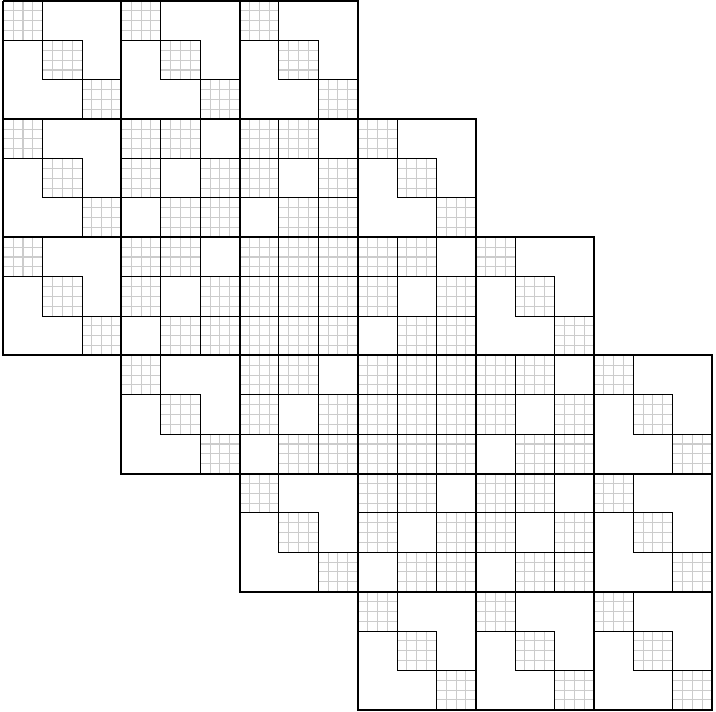}
	\vspace{-0.1cm}
	\caption{Code array corresponding to $\mathcal{C}_{24}(\etab)$ in
	Def.~\ref{def:deterministic_ensemble} with $L = 6$, $w = 3$. } 
	\vspace{-0.5cm}
	\label{fig:mod_braided_code_array}
\end{figure}

\begin{example}
	Fig.~\ref{fig:mod_braided_code_array} shows the (not necessarily
	unique) code array for $L = 6$ and $w=3$, where $\mat{A}^\transpose
	\mat{A}$ is given in \eqref{eq:b_matrix}, which can be regarded as a
	type of braided code.  \demo
\end{example}


Using the structure of $\etab$ in
Def.~\ref{def:deterministic_ensemble}, one can show that the DE
recursion for $\mathcal{C}_n(\etab)$ in \eqref{eq:de_gpc} is
equivalent to \eqref{eq:de_sc_ensemble_mod}. For example, the step in
\eqref{eq:eta_prescription} is the opposite of the simplification in
Ex.~\ref{ex:braided_effective}. The recursion defined by
\eqref{eq:de_sc_ensemble_mod} constitutes an (unconditionally stable)
scalar admissible system as defined in \cite{Yedla2014}. One may thus
use the potential function approach in \cite{Yedla2014} to calculate
decoding thresholds as follows (see also \cite{Jian2014, Zhang2015}). 

\begin{definition}
	\label{def:potential_function}
	The single system potential function is defined as $V_s(x) \define
	\frac{1}{2} x^2 - H(x)$, where $H(x) = \int_0^z h(z) \, \D z$.  In
	order to highlight the dependence of the potential function on the
	channel quality parameter $c$, we write $V_s(x; c)$. 
\end{definition}

\begin{definition}
	\label{def:potential_threshold}
	The potential threshold is defined as
	\begin{align}
		\label{eq:potential_threshold}
		\cp = \sup \{ c \geq 0 \,|\, \min_{x
		\in [0;1]} V_s(x; c) \geq 0 	\}.
	\end{align}
\end{definition}

Using \cite[Lem.~36]{Yedla2014}, we have the following theorem. 

\begin{theorem}
	\label{th:ensemble_performance}
	Let $\etab$ and $\gamma$ be as in
	Def.~\ref{def:deterministic_ensemble}. For any $c < \cp$, there
	exists $w_0 < \infty$ such that for all $w \geq w_0$ and all $L$,
	the DE recursion \eqref{eq:de_gpc} for $\mathcal{C}_n(\etab)$
	converges to the $\vect{0}$ vector. 
\end{theorem}


\begin{remark}
	From Th.~\ref{th:ensemble_performance}, the threshold of
	$\mathcal{C}_n(\etab)$ satisfies $\cthr \geq \cp$ for all $L$ and
	$w$ sufficiently large. One can further show that $\cthr \geq 2\tee
	- 2$ if, additionally, $\tee$ is sufficiently large. The latter
	result was proved in \cite[Lem.~8]{Jian2012} for the
	spatially-coupled ensemble.  It also applies to the deterministic
	braided codes in Def.~\ref{def:deterministic_ensemble}, since the DE
	equations are equivalent. 
\end{remark}

\subsection{Simpler Deterministic Codes}

The curious structure of the code array in
Fig.~\ref{fig:mod_braided_code_array} is due to our attempt of
``reverse-engineering'' the DE equations of the ensemble by means of
the deterministic code construction. This begs the question whether
there exist other deterministic spatially-coupled PCs that exhibit a simpler
structure but still achieve performance guarantees similar to those in
Th.~\ref{th:ensemble_performance}. The most natural candidate appears
to be the extension of the block-wise braided code in
Fig.~\ref{fig:code_arrays}(b) to larger coupling widths. 

\begin{definition}
	\label{def:braided_extended}
	For given $L$ and $w$, let $\gamma = (2w -1)^{-1}$ and let the $L
	\times L$ matrix $\etab'$ be defined by $\eta'_{i,j} =
	\Indicator{|i-j| < w }$. Finally, let $\etab$ be as in
	\eqref{eq:eta_prescription} for $i,j \in [L]$. 
\end{definition}

\begin{example}
	For $w = 2$, $\etab$ in Def.~\ref{def:braided_extended} recovers
	$\etab$ in Ex.~\ref{ex:braided_codes}. \demo
\end{example}

The resulting DE recursion for $\mathcal{C}_n(\etab)$ is neither
equivalent to the ensemble DE recursion nor to the recursion studied
in \cite{Yedla2014}. However, one can still show the following. 

\begin{theorem}
	\label{th:braided_convergence}
	Let $\etab$ and $\gamma$ be as in Def.~\ref{def:braided_extended}.
	For any $c < \cp$, there exists $w_0 < \infty$ such that for all $w
	\geq w_0$ and all $L$, the DE recursion \eqref{eq:de_gpc} for
	$\mathcal{C}_n(\etab)$ converges to the $\vect{0}$ vector. 
\end{theorem}
\begin{proof}
	See the Appendix. 
\end{proof}
	
\section{Potential Threshold Optimization}
\label{sec:potential_threshold_optimization}

In this section, we consider the case where we assign different
erasure-correcting capabilities to the component codes. To that end,
let $\vect{\tau} = (\tau_1, \dots, \tau_{\tmax})^\transpose$ be a probability
vector (i.e., $\vect{1}^\transpose \vect{\tau} = 1$ and $\vect{\tau}
\succeq 0$), where $\tau_{\tee}$ denotes the fraction of CNs at each
position that can correct $\tee$ erasures and $\tmax$ is the maximum
erasure-correcting capability.  We further define the average
erasure-correcting capability as $\bar{\tee} \define \sum_{\tee =
1}^{\tmax} \tee \tau_\tee$.  The assignment can be done either
deterministically if $\tau_\tee d$ is an integer for all
$\tee$, or independently at random according to $\vect{\tau}$. In both
cases, the distribution $\vect{\tau}$ manifests itself in the DE
equation \eqref{eq:de_gpc} by changing the function $h$ defined in
Sec.~\ref{sec:de} to $h(x) = \sum_{\tee = 1}^{\tmax} \tau_\tee
\Psi_{\geq \tee}(c x)$ (see \cite{Haeger2015tit} for details). This
affects the potential function in Def.~\ref{def:potential_function}
and thus also the potential threshold in
Def.~\ref{def:potential_threshold}. In particular, both quantities now
depend on $\vect{\tau}$ and this change is reflected in our notation
by writing $V_s(x; c, \vect{\tau})$ and $\cp(\vect{\tau})$,
respectively. 

\begin{definition}
	\label{def:semi_regular}
	A distribution is said to be regular if $\tau_{\bar{\tee}} = 1$ for
	$\bar{\tee} \in \mathbb{N}$ and semi-regular if
	$\tau_{\floor{\bar{\tee}}} = 1 + \floor{\bar{\tee}} - \bar{\tee}$
	and $\tau_{\floor{\bar{\tee}}+1} = \bar{\tee} - \floor{\bar{\tee}} $
	for $\bar{\tee} \notin \mathbb{N}$.
\end{definition}

\begin{theorem}
	\label{th:regular_optimal}
	For any fixed mean erasure-correcting capability $\bar{\tee} \geq
	2$, a (semi-)regular distribution maximizes the potential threshold
	$\cp(\vect{\tau})$.
\end{theorem}

\begin{proof}
	See the Appendix.
\end{proof}

Th.~\ref{th:regular_optimal} is in contrast to conventional PCs which
typically {\nobreak benefit} from employing component codes with different
strengths. However, Th.~\ref{th:regular_optimal} does not necessarily
imply that there can be no practical value in employing different
component codes also for spatially-coupled PCs. In practice,
quantities such as the coupling width, the component code length, and
the number of decoding iterations are constrained to be finite.
Depending on the severity of these constraints, the potential threshold
may not be a good performance indicator. 

\section{Conclusion}
\label{sec:conclusion}

We studied the asymptotic performance of deterministic
spatially-coupled PCs under iterative decoding. We showed that there
exists a family of deterministic braided codes that performs
asymptotically equivalent to a previously considered spatially-coupled
PC ensemble. There also exists a related but structurally simpler
braided code family that attains essentially the same asymptotic
performance. Lastly, we showed that employing component code mixtures
for spatially-coupled PCs is not beneficial from an asymptotic point
of view. 

\appendix

\section{Proof of Theorem \ref{th:braided_convergence} and
\ref{th:regular_optimal}}
\label{app:proofs}

\emph{Proof of Theorem \ref{th:braided_convergence}.} The recursion of
interest (after removing odd positions due to symmetry as explained in
Ex.~\ref{ex:braided_effective}) is given by $\vect{x}^{(\ell)} =
\vect{h}(\mat{B}' \vect{x}^{(\ell-1)})$, where $\mat{B}' = \gamma
\etab'$ and $\gamma$, $\etab'$ are as in
Def.~\ref{def:braided_extended}. The authors in \cite{Yedla2014} study
the recursion
\begin{align}
	\label{eq:yedla_de}
	\vect{y}^{(\ell)} = \mat{A}^\transpose \vect{f}( \mat{A}
	\vect{g} (\vect{y}^{(\ell-1)} )) = \mat{A}^\transpose
	\vect{f}(\tilde{\vect{y}}^{(\ell)})
\end{align}
for suitable functions $f$, $g$, where $\vect{\tilde{y}}^{(\ell)} =
\mat{A} \vect{g}(\vect{y}^{(\ell-1)})$ is defined implicitly. Since
$h$ is strictly increasing and analytic, we can let both $f = h$ and
$g = h$. For this proof, $\mat{A}$ is assumed to be of size $L \times
L+\tilde{w}-1$ with $A_{i,j} = \tilde{w}^{-1} \Indicator{1 \leq j-i+1
\leq \tilde{w}}$ for $i \in [L]$, $j \in [L+\tilde{w}-1]$, where
$\tilde{w} \define 2w - 1$. The potential function $U_s(x; c) = h(x)x
- H(x) - H(h(x))$ associated with the scalar recursion $x^{(\ell)} =
h(h(x^{(\ell-1)}))$ as defined in \cite[eq.~(4)]{Yedla2014} predicts
the same potential threshold as the one in
Def.~\ref{def:potential_function}. According to
\cite[Lem.~36]{Yedla2014}, the claim in the theorem is thus true for
the recursion \eqref{eq:yedla_de}. To show that it must also be true
for the recursion of interest, we argue as follows. Assume that we
swap the application of $\vect{h}$ and $\mat{B}'$ in the recursion of
interest and then consider ``two iterations at once'' according to
\begin{align}
	\label{eq:de_double}
	\vect{z}^{(\ell)} = \mat{B}' \vect{h}( \mat{B}'
	\vect{h} ( \vect{z}^{(\ell-1)} )) = \mat{B}' \vect{h}(
	\vect{\tilde{z}}^{(\ell)} ).
\end{align}
We claim that \eqref{eq:yedla_de} dominates \eqref{eq:de_double}, in
the sense that $\vect{y}^{(\infty)} = \vect{0}$ implies
$\vect{z}^{(\infty)} = \vect{0}$ (and thus $\vect{x}^{(\infty)} =
\vect{0}$). To see this, observe that $\vect{y}^{(\ell)}$ has length
$L+\tilde{w}-1$, whereas $\vect{\tilde{y}}^{(\ell)}$,
$\vect{z}^{(\ell)}$, and $\vect{\tilde{z}}^{(\ell)}$ have length $L$.
We use $\vect{y}^{(\ell)} = ( (\vect{y}^{(\ell)}_\text{t} )
^\transpose, (\vect{y}^{(\ell)}_\text{c})^\transpose,
(\vect{y}^{(\ell)}_\text{b}) ^\transpose)^\transpose$ to denote the
$w-1$ top, $L$ center, and $w-1$ bottom entries in
$\vect{y}^{(\ell)}$. We want to show that $\vect{y}^{(\ell)}_\text{c}
\succeq \vect{z}^{(\ell)}$ for all $\ell$.  Assume this is true for
$\ell - 1$. This gives the second inequality in 
\begin{align*}
\vect{\tilde{y}}^{(\ell)} = \mat{A} \vect{h}
(\vect{y}^{(\ell-1)}) \succeq \mat{B}' \vect{h}(
\vect{y}_{\text{c}}^{(\ell-1)} ) \succeq \mat{B}' \vect{h}( 
\vect{z}^{(\ell-1)} ) = \vect{\tilde{z}}^{(\ell)}, 
\end{align*}
where the first inequality follows from $\vect{y}^{(\ell-1)}_\text{t},
\vect{y}^{(\ell-1)}_\text{b} \succeq \vect{0}$ (since
$\vect{y}^{(\ell)} \succeq \vect{0}$ for all $\ell$) and the (almost
identical) structure of $\mat{A}$ and $\mat{B}'$.  Observe that we
have $\vect{y}^{(\ell)}_\text{c} = \mat{B}'
\vect{h}(\vect{\tilde{y}}^{(\ell)})$. Also $\vect{z}^{(\ell)} =
\mat{B}' \vect{h}(\vect{\tilde{z}}^{(\ell)})$ and, since we have just
shown that $\vect{\tilde{y}}^{(\ell)} \succeq
\vect{\tilde{z}}^{(\ell)}$, the claim follows by induction on $\ell$.
\qed

\medskip

\emph{Proof of Theorem \ref{th:regular_optimal}.} Using integration by
parts, one may verify that the potential function in
Def.~\ref{def:potential_function} is given by
\begin{align}
	\label{eq:potential_function_def}
	V_s(x; c, \vect{\tau}) = x^2/2 - x + (\bar{\tee} -
	\Loss_{\vect{\tau}}(c x))/c, 
\end{align}
where we defined $\Loss_{\vect{\tau}}(x) \define \sum_{\tee =
1}^{\tmax} \tau_\tee \Loss(\tee, x)$, with $\Loss(\tee, x) \define
\sum_{k=0}^{\tee-1} \Psi_{=k}(x) (\tee - k)$ for $\tee \in
\mathbb{N}$. For any fixed $x \geq 0$, we also define the affine
extension of $\Loss(\tee, x)$ for $\tee \in [1, \infty)$ as
\begin{align}
	\label{eq:loss_affine_extension}
	\Loss(\tee, x) = \Loss(\floor{\tee}, x) +
	(\Loss(\ceil{\tee}, x) - \Loss(\floor{\tee}, x))
	(\tee - \floor{\tee}).
\end{align}
The proof relies on the fact that $\Loss(\tee, x)$ is convex in $\tee
\in [1, \infty)$ for any $x\geq 0$. Indeed, since $\Loss(\tee, x)$ is the
affine extension of a discrete function, it suffices to show that for
$\tee \in \{2, 3, \dots\}$,
\begin{align}
	\Loss(\tee-1, x) + 
	\Loss(\tee+1, x) &= 
	2 \Loss(\tee, x) + \Psi_{=\tee}(x) \\
	&\geq 2 \Loss(\tee, x), 
\end{align}
since $\Psi_{=t}(x) \geq 0$ with equality if and only if $x = 0$. As a
consequence, for any distribution $\vect{\tau}$ with average
erasure-correcting capability $\bar{\tee}$ and any $x \geq 0$, we have 
\begin{align}
	\label{eq:loss_convexity}
	\Loss_{\vect{\tau}} (x) \geq \Loss(\bar{\tee}, x) =
\Loss_{\vect{\tau}_\text{reg}}(x),
\end{align}
where $\vect{\tau}_\text{reg}$ denotes the (semi-)regular distribution
in Def.~\ref{def:semi_regular}.

Now, let $\vect{\tee} \define (1, 2, \dots, \tmax)^\transpose$ and consider
\begin{equation}
	\label{eq:potential_threshold_program}
	\begin{aligned}
		\underset{\vect{\tau} \in \mathcal{T}}{\text{max}} \,\,
		\cp(\vect{\tau}) 
		 \,\,
	 \text{subject to} \,\, 
		\vect{\tee}^\transpose \vect{\tau} = \bar{\tee},
	\end{aligned}
\end{equation}
where $\mathcal{T} = \{ \vect{\tau} \in \mathbb{R}^{\tmax} \,|\,
\vect{1}^\transpose \vect{\tau} = 1,  \vect{\tau} \succeq 0 \}$. This
can be equivalently written in epigraph form as
\begin{equation}
	\label{eq:epigraph_form}
	\begin{aligned}
		\underset{c \geq 0, \vect{\tau} \in \mathcal{T}}{\text{max}} 
		\,\, c \,\, \text{subject to} \,\, c \leq \cp(\vect{\tau}),
		\,
		\vect{\tee}^\transpose  \vect{\tau} = \bar{\tee}.
	\end{aligned}
\end{equation}
According to \eqref{eq:potential_threshold}, the first constraint in
\eqref{eq:epigraph_form} is equivalent to $V_s(x; c, \vect{\tau}) \geq
0$ for $x \in [0; 1]$. Assume that \eqref{eq:epigraph_form} is
maximized by some $(c^*, \vect{\tau}^*)$. Then, for all $x \in [0;
1]$, we have
\begin{align}
	0 \leq V_s(x; c^*, \vect{\tau}^*) \leq V_s(x; c^*,
	\vect{\tau}_\text{reg}),
\end{align}
where the last inequality follows from
\eqref{eq:potential_function_def} and \eqref{eq:loss_convexity}.
Thus, the (semi-)regular distribution $\vect{\tau}_\text{reg}$ is
feasible and attains (at least) the same threshold value $c^*$. \qed



\end{document}